\documentclass[prl,aps,superscriptaddress,floatfix,reprint]{revtex4-2}
\usepackage{amsthm}
\usepackage{amsmath}
\usepackage{amssymb}
\usepackage{amsfonts}
\usepackage{epsfig}
\usepackage{color}
\usepackage{bm}
\usepackage{times}
\usepackage[colorlinks,linkcolor=blue,anchorcolor=blue,citecolor=blue,urlcolor=blue]{hyperref}
\usepackage{braket}
\usepackage{graphicx}
\usepackage{setspace}
\usepackage{subfigure}
\usepackage{enumerate}
\theoremstyle{definition}

\newtheorem*{definition*}{Definition}
\theoremstyle{plain}

\newtheorem{proposition}{Proposition}
\newtheorem*{theorem*}{Theorem}
\newtheorem*{corollary*}{Corollary}
\newtheorem*{lemma*}{Lemma}
\newtheorem*{proposition*}{Proposition}

\NewDocumentEnvironment{variant}{O{theorem} D(){} m}
{\addtocounter{#1}{-1}%
\expandafter\renewcommand\csname the#1\endcsname{\ref{#3}$'$}%
\begin{#1}[#2]}
{\end{#1}}
\NewDocumentEnvironment{appthm}{O{theorem} D(){} m}
{\addtocounter{#1}{-1}%
\expandafter\renewcommand\csname the#1\endcsname{\ref{#3}}%
\begin{#1}[#2]}
{\end{#1}}

\begin{document}
\newcommand{\papertitle}{Operation Mpemba effect: Breakdown of resource-Markovianity of free dynamics}
\title{\papertitle}

    \author{Tian-Ren Jin}
    \affiliation{Institute of Physics, Chinese Academy of Sciences, Beijing 100190, China}
    \affiliation{School of Physical Sciences, University of Chinese Academy of Sciences, Beijing 100049, China}

    \author{Yu-Ran Zhang}
    \email{yuranzhang@scut.edu.cn}
    \affiliation{School of Physics and Optoelectronics, South China University of Technology, Guangzhou 510640, China}

    \author{Heng Fan}
    \email{hfan@iphy.ac.cn}
    \affiliation{Institute of Physics, Chinese Academy of Sciences, Beijing 100190, China}
    \affiliation{School of Physical Sciences, University of Chinese Academy of Sciences, Beijing 100049, China}
    \affiliation{Beijing Academy of Quantum Information Sciences, Beijing 100193, China}
    \affiliation{Hefei National Laboratory, Hefei 230088, China}
    \affiliation{Beijing Key Laboratory of Advanced Quantum Technology, Beijing 100190, China}

\begin{abstract}
    The Mpemba effect refers to faster relaxation of states that are initially farther from equilibrium, yet its characterization is often tied to a chosen distance or resource measure.
    We introduce resource-Markovianity, an extended concept of quantum Markovianity to quantum resource theories, and formulate the resource Mpemba effect operationally as the breaking of resource-Markovianity by a relaxation operation.
    This yields a measure-independent operational characterization of resource Mpemba effects in general resource theories, together with quantitative characterizations based on resource-non-Markovianity measures.
    We illustrate the framework with the Mpemba effect for distinguishability of states, due to its relation to quantum Markovianity, and with the thermomajorization Mpemba effect from an operational perspective.
    These results reveal a deep interplay between quantum resources, non-Markovianity, and the Mpemba effect.
\end{abstract}
\maketitle
    \emph{Introduction}---%
    Mpemba effect describes a counterintuitive phenomenon where a hotter system cools down faster than a cooler one~\cite{mpemba1969cool,burridge2016questioning,bechhoefer2021fresh}. 
    In recent years, it has been developed into a broadly studied nonequilibrium relaxation phenomenon, with extensions far beyond its original hot-water formulation~\cite{lu2017nonequilibrium,PhysRevX.9.021060,ibanez2024heating,PhysRevLett.134.107101}. 
    Representative examples include granular and molecular gases~\cite{PhysRevLett.119.148001,PhysRevE.99.060901,santos2020mpemba,PhysRevE.105.054140}, spin glasses~\cite{baity2019mpemba}, colloidal systems~\cite{kumar2020exponentially,kumar2022anomalous,PhysRevLett.129.138002}, and engineered relaxation protocols~\cite{PhysRevLett.131.017101,PhysRevLett.132.117102}.
    Recently, this effect has also been observed in quantum systems, attracting considerable attention~\cite{ares2025quantum,PhysRevB.100.125102,PhysRevLett.127.060401,PhysRevA.106.012207,ares2023entanglement,PhysRevLett.131.080402,PhysRevLett.133.010401,PhysRevLett.133.010402,PhysRevLett.133.010403,PhysRevLett.133.136302,murciano2024entanglement,PhysRevLett.133.140405,PhysRevB.111.104312,PhysRevB.111.104506,PhysRevB.111.125404,PhysRevResearch.6.033330,PhysRevLett.134.220403,5d6p-8d1b,9f6l-d766,yzjd-pk8h,xu2025observation}.
    In the context of the quantum Mpemba effect, two quantum states $\rho_1(0)$ and $\rho_2(0)$ are prepared at different distances from an equilibrium state $\rho^*$, with $d(\rho_1(0),\rho^*) > d(\rho_2(0),\rho^*)$. 
    Under relaxation dynamics, the corresponding distance curves may cross at a time $t_c$, such that $d(\rho_1(t),\rho^*) \leq d(\rho_2(t),\rho^*)$ for $t>t_c$.
    The quantum Mpemba effect has been explored beyond thermalizing dynamics for instance in symmetry-restoration settings~\cite{ares2023entanglement,PhysRevLett.133.140405,ares2023lack,rylands2024dynamical,chalas2024multiple} and has also been extended to various quantum resources within the framework of quantum resource theory~\cite{rbt4-psfd,aditya2025mpemba}.

    For the resource Mpemba effect, one considers two resourceful states $\rho_1(0)$ and $\rho_2(0)$ with different amount of specific resources, measured by $M(\rho)$, such that $M(\rho_1(0)) > M(\rho_2(0))$. 
    Since $M(\rho)$ is monotonic under free operations $\mathcal{N}_{\mathrm{free}}$, $M(\mathcal{N}_{\mathrm{free}}(\rho)) \leq M(\rho)$, free relaxation drives resourceful states toward free states.
    The resource Mpemba effect occurs when more resourceful states relax faster than the less resourceful ones, i.e., there exists a time $t_c$ such that $M(\rho_1(t)) < M(\rho_2(t))$ for $t > t_c$.

    A fundamental challenge in characterizing the Mpemba effect lies in its dependence on the choice of distance measure $d(\rho_1,\rho_2)$~\cite{PhysRevE.108.024902}, since physical phenomena should ideally be independent of the specific chosen measure. 
    In classical thermodynamics, thermomajorization theory has been proposed to solve this problem~\cite{PhysRevLett.134.107101}. 
    In quantum resource theories, measure-dependent resource Mpemba effects remain physically meaningful because resource measures quantify specific manipulation tasks, and in reversible theories, the ambiguity disappears since convertibility is governed by a unique independent measure.
    However, in general resource theories there is still no universal measure-independent characterization of the Mpemba effect.

    Recently, the Pontus-Mpemba effect has been proposed as a ``useful'' version of Mpemba effect, demonstrating that for a fixed initial state, a heating stage before relaxation can lead to faster relaxation than direct relaxation~\cite{hhgj-89gj}. 
    As an initialization procedure, the heating stage plays a central role in the Pontus-type protocol, motivating an operational viewpoint of the Mpemba effect. 
    In addition, free convertibility provides a natural way to organize resource-monotonic states in resource theories. 
    For instance, thermomajorization can be understood as the order relation induced by free convertibility in classical thermodynamics~\cite{sagawa2022entropy}. 
    Recent progress on generalized quantum Stein's lemma has also provided a framework for constructing reversible resource theories~\cite{hayashi2025generalized}, further motivating the use of free operations to organize resource-monotonic states.

\begin{figure*}
    \centering
    \includegraphics[width=\textwidth]{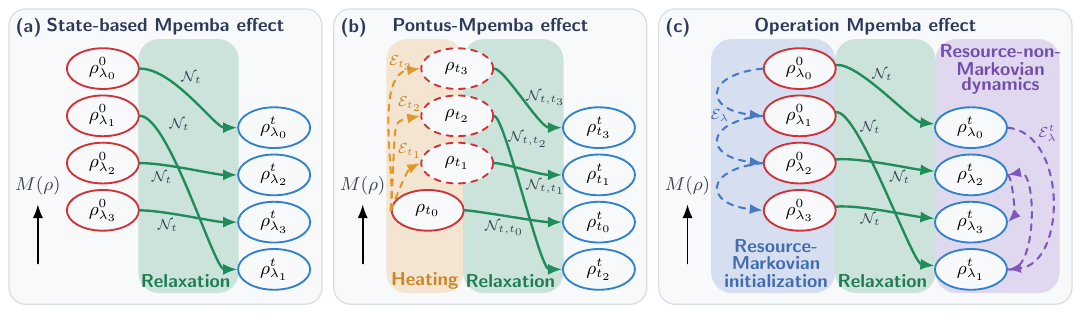}
    \caption{
        Schematic comparison of the state-based, Pontus, and operation perspectives on resource Mpemba effect.
        The vertical axis denotes the resource measure $M(\rho)$, red circles denote initial resourceful states, green arrows denote free relaxations, and blue circles denote the corresponding output states.
        (a) State-based Mpemba effect.
        Different initial states are compared under the same relaxation dynamics, and the more resourceful initial state $\rho_{\lambda_1}^0$ relaxes to a less resourceful state than the less resourceful initial states $\rho_{\lambda_2}^0$ and $\rho_{\lambda_3}^0$.
        (b) Pontus-Mpemba effect.
        A two-step protocol consisting of heating stage followed by relaxation is compared with direct relaxation from the same initial state $\rho_{t_0}$.
        (c) Operation Mpemba effect.
        A relaxation operation $\mathcal{N}_t$ breaks the resource-Markovianity of an initialization dynamics $\mathcal{E}_{\lambda}$, turning it into the induced dynamics $\mathcal{E}_{\lambda}^t=\mathcal{N}_t\circ\mathcal{E}_{\lambda}$.
    }
    \label{figure1}
\end{figure*}

    In this Letter, we investigate resource-monotonic initial states generated from a fixed state by a free initialization dynamics. 
    We call such a dynamics resource-Markovian if the resource measure of the evolved states decreases monotonically along the dynamics. 
    A relaxation dynamics then exhibits an operation Mpemba effect when it breaks the resource-Markovianity of the initialization dynamics. 
    This yields a measure-independent operational characterization of the Mpemba effect in general resource theories. 
    In reversible resource theories, the state-based formulation acquires an operational interpretation. 
    Moreover, the operation Mpemba effect is directly linked to resource-Markovianity of free dynamics, which enables a quantitative characterization of the Mpemba effect through the resource-non-Markovianity quantifiers.
    By illustrating the framework with distinguishability, quantified by the trace distance, and with the thermomajorization Mpemba effect from an operational perspective, we show how Mpemba effects can be characterized through resource backflow in general resource theories.

    \emph{Resource-Markovianity}---%
    For the resource Mpemba effect, the resource measure $M(\rho)$ takes the role of the distance measure $d(\rho,\rho^*)$ in the conventional Mpemba effect. 
    The initial states $\rho_{\lambda}$ for the Mpemba effect are resource-monotonic on the parameter $\lambda$, i.e., $M(\rho_{\lambda_1}) \geq M(\rho_{\lambda_2})$ for $\lambda_1 < \lambda_2$. 
    Resource measures are monotonic under the action of free operations $\mathcal{N}_{\mathrm{free}}$, which reflects the fact that resources cannot be freely generated. 
    In resource manipulation tasks, one may expect that the more resourceful state can be freely converted into less resourceful states, i.e., the second law of resource theories. 
    The second law holds with respect to a unique resource measure $M$, i.e., for any states $\rho_1$ and $\rho_2$ such that $M(\rho_1)\geq M(\rho_2)$, there exists a free operation $\mathcal{F}$ with $\rho_2=\mathcal{F}(\rho_1)$~\cite{PhysRevLett.115.070503}.
    This motivates us to define the resource-monotonic states $\rho_{\lambda}$ by free operations $\mathcal{E}_{\lambda_2,\lambda_1}$, i.e., $\rho_{\lambda_2} = \mathcal{E}_{\lambda_2,\lambda_1}(\rho_{\lambda_1})$.

    The free dynamics $\mathcal{E}_{\lambda}=\mathcal{E}_{\lambda,0}$, induced by $\mathcal{E}_{\lambda_2,\lambda_1}$, satisfies, for $\lambda_1<\lambda_2$ and any state $\rho$,
    \begin{equation}
        M(\mathcal{E}_{\lambda_1}(\rho)) \geq M(\mathcal{E}_{\lambda_2}(\rho)).
    \end{equation}
    Such a dynamics with this property is called \textit{resource-Markovian}. 
    If the resource-Markovian dynamics $\mathcal{E}_{\lambda}$ is invertible, it is called \textit{free-divisible}, since the intermediate map
    \begin{equation}
        \mathcal{E}_{\lambda_2,\lambda_1} = \mathcal{E}_{\lambda_2} \circ \mathcal{E}_{\lambda_1}^{-1},
    \end{equation}
    exists and is free for $\lambda_1 < \lambda_2$.
    The resource-Markovianity refers to the resource flow from system to environment along the free dynamics $\mathcal{E}_{\lambda}$.
    These notions are directly analogous to Markovianity~\cite{PhysRevLett.103.210401} and (C)P-divisible quantum dynamics~\cite{wolf2008dividing,PhysRevLett.103.210401,PhysRevLett.105.050403,rivas2014quantum,RevModPhys.88.021002}, where the trace distance $D(\rho,\sigma)$ plays the role of the resource measure and the free operations are replaced with (C)PTP channels.

    In analogy to the Markovianity of the quantum dynamical maps, we consider the change rate of the resource measure,
    \begin{equation}
        \sigma(\mathcal{K}_{\lambda},\rho) = \frac{\mathrm{d}}{\mathrm{d} \lambda} M(\mathcal{K}_{\lambda}(\rho))
    \end{equation}
    to characterize the resource-Markovianity of a dynamics $\mathcal{K}_{\lambda}$.
    We then define
    \begin{equation} \label{eq: resource_Markovianity}
        N(\mathcal{K}_{\lambda}) = \max_{\rho} \int_{\sigma>0} \!\!\mathrm{d}\lambda \; \sigma(\mathcal{K}_{\lambda}, \rho) \geq 0,
    \end{equation}
    which quantifies the resource backflow for the dynamics $\mathcal{K}_{\lambda}$, in analogy with non-Markovianity measures for quantum dynamics~\cite{PhysRevLett.103.210401}.
    If the operations are resource-Markovian, then $\sigma(\mathcal{K}_{\lambda},\rho)\leq 0$, for all states $\rho$, and consequently $N(\mathcal{K}_{\lambda}) = 0$. 
    Otherwise, the operations are resource-non-Markovian, if $N(\mathcal{K}_{\lambda}) > 0$.
    In particular, for the discrete parameter $\lambda_i \ (i = 1,2,\ldots)$, denoting $[x]_+=\max\{x,0\}$, the measure is
    \begin{equation}
        N(\mathcal{K}) = \max_{\rho} \sum_i \left[M(\mathcal{K}_{\lambda_{i+1}}(\rho)) - M(\mathcal{K}_{\lambda_i}(\rho))\right]_+.
    \end{equation}
    
    \emph{Operation Mpemba effect}---%
    With the notion of resource-Markovianity, we formulate the Mpemba effect from an operational viewpoint.
    Given a state $\rho$, we consider the family of the initial states $\rho_{\lambda} = \mathcal{E}_{\lambda}(\rho)$, where $\mathcal{E}_{\lambda}$ is resource-Markovian. 
    Due to resource-Markovianity, the family of initial states $\rho_{\lambda}$ is monotonic on the parameter $\lambda$ under the resource measure $M(\rho)$, i.e., $M(\rho_{\lambda_1}) \geq M(\rho_{\lambda_2})$ for $\lambda_1 < \lambda_2$. 
    The initial states evolve under a free operation $\mathcal{N}_t$ of relaxation dynamics. 
    If the resource Mpemba effect occurs at the time $t$, then there exist two parameters $\lambda_1 < \lambda_2$ such that
    \begin{equation}
        M(\mathcal{N}_t(\rho_{\lambda_1})) < M(\mathcal{N}_{t}(\rho_{\lambda_2})).
    \end{equation}
    This implies that the induced dynamics $\mathcal{N}_t \circ \mathcal{E}_{\lambda}$ is no longer resource-Markovian in the parameter $\lambda$.

    Therefore, the operational formulation of the Mpemba effect can be understood as the free relaxation operation $\mathcal{N}$ breaking the resource-Markovianity given some resource-Markovian dynamics $\mathcal{E}_{\lambda}$. 
    Here follows the implication.
    \begin{proposition}
        If the operation Mpemba effect occurs for $\mathcal{N}_t$, there exists a family of initial states $\rho_{\lambda}$ that exhibits a resource Mpemba effect under $\mathcal{N}_t$.
    \end{proposition}
    \begin{proof}
        Assume the operation Mpemba effect occurs for $\mathcal{N}_t$ on the resource-Markovian dynamics $\mathcal{E}_{\lambda}$. 
        Then, the induced dynamics $\mathcal{N}_t \circ \mathcal{E}_{\lambda}$ is no longer resource-Markovian in $\lambda$. Consequently, there exist two parameters $\lambda_1 < \lambda_2$ and some state $\rho_0$, such that the resource measure of the evolved states increases, i.e., $M(\mathcal{N}_t \circ \mathcal{E}_{\lambda_1}(\rho_{0})) < M(\mathcal{N}_t \circ \mathcal{E}_{\lambda_2}(\rho_{0}))$.
        This shows that the resource-monotonic initial states $\rho_{\lambda} = \mathcal{E}_{\lambda}(\rho_0)$ exhibits a resource Mpemba effect under $\mathcal{N}_t$.
    \end{proof}
    Thus, the operational formulation is a special case of the state-based formulation. 
    Conversely, in reversible resource theories, the state-based formulation also acquires an operational interpretation. 
    If the reversible resource theory further admits one-shot free convertibility, the operational and state-based formulations are equivalent.
    See Supplemental Material~\ref{app: reversibility}, for the proofs.
    Figures~\ref{figure1}(a)--(c) distinguish the state-based, Pontus, and operation Mpemba effects. 
    Our formulation is inspired by Pontus-Mpemba, but it concerns the conventional Mpemba setting, where different initial states are prepared by a resource-Markovian initialization dynamics and then relaxed by the same operation.

    Without loss of generality, we assume the existence of the inverses for both $\mathcal{E}_{\lambda}$ and $\mathcal{N}$. 
    Otherwise, when $\mathcal{N}$ is non-invertible, one may consider the family $(1-\epsilon)\mathcal{N} + \epsilon\mathcal{I}$, which approaches $\mathcal{N}$ as $\epsilon \to 0$. 
    Given this assumption, the induced operations $\mathcal{N} \circ \mathcal{E}_{\lambda}$ are also invertible, and the intermediate maps are
    \begin{equation}
        \tilde{\mathcal{E}}_{\lambda_2,\lambda_1} = \mathcal{N} \circ \mathcal{E}_{\lambda_2,\lambda_1} \circ \mathcal{N}^{-1}.
    \end{equation}
    If no Mpemba effect occurs, $\tilde{\mathcal{E}}_{\lambda_2,\lambda_1}$ is free for every $\lambda_2>\lambda_1$. 
    Therefore, the operation $\mathcal{N}$ has no operation Mpemba effect if and only if the free set $\mathbb{F}$ of operations is invariant under the adjoint action
    \begin{equation}
        \mathcal{N} \circ \mathbb{F} \circ \mathcal{N}^{-1} \subset \mathbb{F}.
    \end{equation}
    For the corresponding condition in Lindblad-generated free dynamics, see the discussion in Supplemental Material~\ref{app: master}.

    With the measure of resource-Markovianity $N(\mathcal{K}_{\lambda})$, Eq.~\eqref{eq: resource_Markovianity}, we can characterize the operation Mpemba effect of $\mathcal{N}$ by considering the maximum breaking of resource-Markovianity
    \begin{equation}
        P(\mathcal{N}) = \max_{N(\mathcal{E}_{\lambda})=0} N(\mathcal{N} \circ \mathcal{E}_{\lambda}) > 0.
    \end{equation}
    For the state Mpemba effect, replacing $\mathcal{E}_{\lambda}(\rho)$ with the resource-monotonic states $\rho_{\lambda}$, the effect is quantified by
    \begin{equation}
        P(\mathcal{N}) = \max_{\rho_{\lambda}} \int_{\sigma>0} \! \! \mathrm{d}\lambda \; \sigma(\mathcal{N}(\rho_{\lambda})) > 0,
    \end{equation}
    where the maximization is over all resource monotonic states and $\sigma(\rho_{\lambda}) = \frac{\mathrm{d}}{\mathrm{d} \lambda} M(\rho_{\lambda})$.
    These quantities characterize the maximum resource backflow induced by the relaxation operation $\mathcal{N}$ on the resource-Markovian dynamics or resource-monotonic states.
    In the following, we will demonstrate the operational framework of the Mpemba effect with these witnesses in examples about the Mpemba effect for the distinguishability and thermomajorization.

\begin{figure*}
    \centering
    \includegraphics[width=\textwidth]{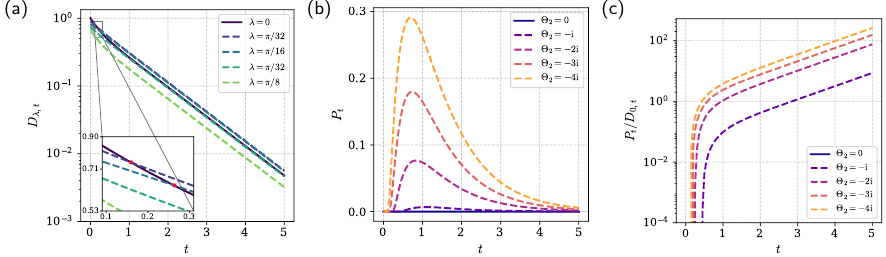}
    \caption{
        Mpemba effect for distinguishability. 
        The state $\rho_1,\rho_2$ are represented by the Bloch vectors $\vec{r}^{(i)} = (r_1^{(i)},r_2^{(i)},r_3^{(i)})^T$ with $r_3^{(i)} = 0$ for $i = 1,2$.
        The trace distance is $D(\rho_1,\rho_2) = \frac{1}{2} \|\vec{r}^{(12)}\|_1$, where $\vec{r}^{(12)} = \vec{r}^{(1)} - \vec{r}^{(2)} = (r_1^{(12)},r_2^{(12)})^{T}$.
        On the two-dimensional vector $\vec{r}^{(12)}$, the initialization and relaxation dynamics are represented as $\exp(\Theta \lambda)$ and $\exp(\Omega t)$ respectively, where $\Theta = \vec{\Theta} \cdot \hat{\bm{\sigma}} + \Theta_0$ with $\vec{\Theta} = (\Theta_1, \Theta_2, \Theta_3)^T$ and $\Omega= \vec{\Omega} \cdot \hat{\bm{\sigma}} + \Omega_0$ with $\vec{\Omega} = (\Omega_1, \Omega_2, \Omega_3)^T$.
        (a) Dynamics of the trace distance $D_{\lambda,t} = D\left(\mathcal{N}_t \circ \mathcal{E}_{\lambda}(\ket{+}),\mathcal{N}_t \circ \mathcal{E}_{\lambda}(\ket{-})\right)$ for different initial parameters $\lambda$. 
        The initial Bloch vector difference $\vec{r}^{(12)} = (2,0)^T$, the parameters of the initialization dynamics are $\Theta_0 = -1, \vec{\Theta} = (0,-4 \mathrm{i},0)^T$, and the relaxation operation are $\Omega_0 = -2, \vec{\Omega} = (1,0,0)^T$.
        Thus, the inner product $|\vec{l} \cdot \vec{\Theta}| = 8$, indicating the occurrence of the operation Mpemba effect for distinguishability. 
        Here, $\vec{l}$ is the eigenvector of adjoint action $\mathrm{ad}_{\Omega}$ of $\Omega$ with positive eigenvalue.
        The Mpemba effect for distinguishability witnessed by trace distance $D_{\lambda,t}$ occurs between $\lambda = 0$ and $\lambda = \pi/32$ or $\lambda = \pi/16$. 
        (b) Dynamics of total positive change $P_{t} = \int_{\sigma > 0} \!\!\mathrm{d}\lambda \; \sigma$, where $\sigma = \mathrm{d} D_{\lambda,t}/\mathrm{d} \lambda$. 
        Parameters of the initialization and relaxation dynamics are the same for (a) except for the varying $\Theta_2$.
        Consequently, the inner product $|\vec{l} \cdot \vec{\Theta}| = 2 |\Theta_2|$ also varies. 
        The nonzero total positive changes $P_t$ verify the occurrence of the Mpemba effect when $\vec{l} \cdot \vec{\Theta} \neq 0$. 
        (c) Dynamics of normalized total positive change $P_{t}/D_{0,t}$. 
        Its increasing indicates that the decreasing total positive change is due to decreasing trace distance.
    }
    \label{figure2}
\end{figure*}

    \emph{Mpemba Effect for distinguishability}---%
    The operation Mpemba effect is defined as the breakdown of the resource-Markovianity for some resource-Markovian dynamics. 
    The quantum Markovianity is a special case of resource-Markovianity, where the free set denotes the whole set of positive and trace-preserving (PTP) quantum channels. 
    For this case, the resource monotonicity is the trace distance between two quantum states $\rho$ and $\sigma$, $D(\rho,\sigma) = \frac{1}{2} \|\rho - \sigma\|_1$, which measures the distinguishability of quantum states~\cite{761271}. 
    Therefore, breaking the quantum Markovianity indicates the occurrence of Mpemba effect for distinguishability.

    Consider the quantum Markovian dynamics $\Phi_{\lambda}$ and the relaxation quantum channel $\mathcal{N}_t$. 
    If the relaxation channel $\mathcal{N}_t$ breaks the Markovianity of the dynamics $\Phi_{\lambda}$ at time $t_c$, the induced dynamics $\mathcal{N}_t \circ \Phi_{\lambda}$ is non-Markovian with $\lambda$. 
    Therefore, there exist two states $\rho$ and $\sigma$, which are initialized as $\rho_{\lambda} = \Phi_{\lambda}(\rho)$ and $\sigma_{\lambda} = \Phi_{\lambda}(\sigma)$, such that for the time $t > t_c$, the trace distance between the relaxed states $\mathcal{N}_t(\rho_{\lambda})$ and $\mathcal{N}_t(\sigma_{\lambda})$ increases at some parameter $\lambda$, i.e., there exist two parameters $\lambda_1 < \lambda_2$,
    \begin{equation}
        D(\mathcal{N}_t(\rho_{\lambda_1}),\mathcal{N}_t(\sigma_{\lambda_1})) < D(\mathcal{N}_t(\rho_{\lambda_2}),\mathcal{N}_t(\sigma_{\lambda_2})).
    \end{equation}
    This indicates that the relaxation operation $\mathcal{N}_t$ exhibits the Mpemba effect for distinguishability.

    We consider a two-level system as an example.
    Any two-level state can be represented as $\rho = \frac{1}{2} (\hat{I} + \sum_{i=1}^3 r_i \hat{\sigma}_i)$, where $\vec{r} = (r_1,r_2,r_3)$ is the Bloch vector. 
    The trace distance between two states $\rho_1$ and $\rho_2$ is $D(\rho_1,\rho_2) = \frac{1}{2} |\vec{r}^{(1)} - \vec{r}^{(2)}|$.
    The Lindblad equation governs the quantum dynamics of the two-level system as
    \begin{equation}
        \frac{\mathrm{d}}{\mathrm{d} t} \rho = -\mathrm{i} [\hat{H}, \rho] + \sum_{i,j} C_{ij} (\hat{\sigma}_i \rho \hat{\sigma}_j - \frac{1}{2} \{\hat{\sigma}_j \hat{\sigma}_i, \rho\}),
    \end{equation}
    where $\hat{H} = \sum_i h_i \hat{\sigma}_i$, and $C_{ij}$ is a Hermitian matrix. 
    It can be represented as the equation of the Bloch vector, ${\mathrm{d} \vec{r}}/{\mathrm{d} t} = \Omega \vec{r} + \vec{b}$, where $\Omega$ is a matrix determined by $h_i$ and $C_{ij}$. 
    The trace distance only depends on the difference $\vec{r}^{(12)} = \vec{r}^{(1)} - \vec{r}^{(2)}$, which satisfies ${\mathrm{d}\vec{r}^{(12)}}/{\mathrm{d} t} = \Omega \vec{r}^{(12)}$.

    Denote the action of the relaxation operation $\mathcal{N}_t$ as $\exp(\Omega t)$, and the initialization dynamics $\mathcal{E}_{\lambda}$ as $\exp(\Theta \lambda)$, see Supplemental Material~\ref{app: distinguishability}, for explicit expressions.
    For simplicity, we set $r_3^{(12)} = 0$, and the Lindblad matrix $\Theta$ reduces to $\Theta = \vec{\Theta} \cdot \hat{\bm{\sigma}} + \Theta_0$, where $\vec{\Theta} = (\Theta_1, \Theta_2, \Theta_3)^T$, and
    The Markovian dynamics $\mathcal{E}_{\lambda}$ satisfies $(\Theta_1^2 + \Theta_3^2)^{1/2} \leq \Theta_0$. 
    The relaxation dynamics $\mathcal{N}_t$ acts on $\vec{\Theta}$ as $\mathrm{ad}_{\Omega} = \vec{\Omega} \cdot \hat{\bm{J}} + \Omega_0$, where
    We consider the case $\Omega_2 = 0$. 
    If the eigenvector $\vec{l}$ of $\mathrm{ad}_{\Omega}$ with a positive eigenvalue is not orthogonal to $\vec{\Theta}$,  the value of $(\Theta_1^2 + \Theta_3^2)^{1/2}$ increases exponentially, and the Mpemba effect occurs.

    We calculate the dynamics of the trace distance $D_{\mathcal{N}_t \circ \mathcal{E}_{\lambda}}(\ket{+},\ket{-})$ for different initialization parameters $\lambda$.
    (See Supplemental Material~\ref{app: distinguishability}, for more details.)
    In Fig.~\ref{figure2}(a), the Mpemba effect for distinguishability occurs for the relaxation operation $\mathcal{N}_t$ between $\mathcal{E}_0$ and $\mathcal{E}_{\pi/32}$ or $\mathcal{E}_{\pi/16}$. In Fig.~\ref{figure2}(b), the total positive change $N_{t}(\mathcal{E}_{\lambda})$ of the trace distance $D_{\mathcal{N}_t \circ \mathcal{E}_{\lambda}}(\ket{+},\ket{-})$ first increases, indicating the occurrence of the Mpemba effect and then decreases to zero. 
    In Fig.~\ref{figure2}(c), the normalized total positive change grows as the time increases, indicating the occurrence of the Mpemba effect at longer time stage.
    This example shows that the operation Mpemba effect for distinguishability is precisely the breakdown of Markovianity under relaxation. Standard non-Markovianity quantifiers therefore directly characterize this effect, providing a direct link between the operational framework of Mpemba effect and quantum non-Markovianity.

\begin{figure*}[t]
    \centering
    \includegraphics[width=\textwidth]{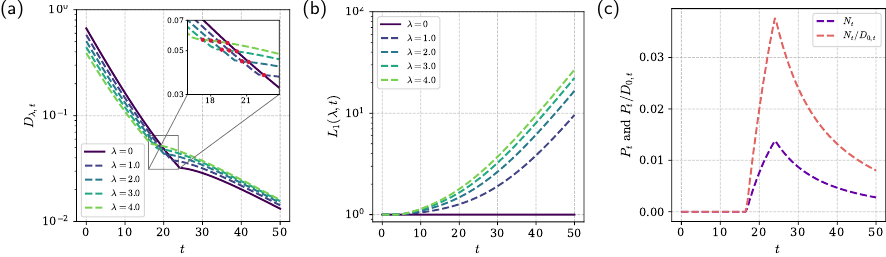}
    \caption{
        Thermomajorization Mpemba effect from the perspective of operation Mpemba effect. 
        The states $\vec{p}(\lambda)$ are initialized from $\vec{p} = (1,0,0)^T$ by initialization dynamics ${\mathrm{d} \vec{p}}/{\mathrm{d}\lambda}  = \hat{A} \vec{p}$ with the symmetric matrix $\hat{A} = [a_{nm}]$ parametrized by $(a_{01},a_{12},a_{02}) \approx (0.010, 0.086, 0.004), a_{nn} = \sum_{m\neq n} a_{mn}$, and then relaxed to $\vec{p}(\lambda,t)$ by the relaxation dynamics ${\mathrm{d} \vec{p}}/{\mathrm{d}t} = \hat{W} \vec{p}$ with the symmetric matrix $\hat{W} = [w_{nm}]$ parametrized by $(w_{01},w_{12},w_{02}) \approx (0.032, 0.066, 0.002), w_{nn} = \sum_{m\neq n} w_{mn}$.
        (a) Dynamics of the detector $D_{\lambda,t} = \max_n [p_n(\lambda,t)/\pi_n]$ of thermomajorization Mpemba effect. 
        The thermomajorization Mpemba effect first occur after the time $t \approx 17$. 
        (b) Dynamics of the $L_1(\lambda,t) = \Vert \hat{R}_t \hat{S}_{\lambda} \hat{R}_t^{-1} \Vert_1$ of the intermediate matrix, which detects the operation Mpemba effect. 
        The operation Mpemba effect first occur at time $t \approx 5$. 
        (c) Dynamics of total positive changes $P_{t} = \int_{\sigma > 0} \!\!\mathrm{d}\lambda \; \sigma$ and the normalization $P_{t}/D_{0,t}$.  
    }
    \label{figure3}
\end{figure*}

    \emph{Thermomajorization Mpemba effect}---%
    In Ref.~\cite{PhysRevLett.134.107101}, the thermomajorization Mpemba effect is introduced to unify the classical thermodynamic Mpemba effect across all monotone measures, thereby providing a measure-independent characterization in classical thermodynamics.
    Let $\vec{\pi} = (\pi_1,\dots,\pi_d)^T$ be a Gibbs thermal state. 
    The classical probability distribution $\vec{p}$ thermomajorizes $\vec{q}$, denoted as $\vec{q} \prec_{\vec{\pi}} \vec{p}$, if the Lorenz curve of $(\vec{p},\vec{\pi})$ lies above that of $(\vec{q},\vec{\pi})$.
    The thermomajorization order $\prec_{\vec{\pi}}$ can be viewed as a measure of the closeness of probability $\vec{p}$ to $\vec{\pi}$, and $\vec{\pi} \prec_{\vec{\pi}} \vec{p}$ for any probability $\vec{p}$~\cite{sagawa2022entropy}.
    For two probability distributions $\vec{p}_t$ and $\vec{q}_t$ relaxed in time $t$, the thermomajorization Mpemba effect occurs at time $t_c$, if for $t < t_c$, $\vec{q}_t \prec_{\vec{\pi}} \vec{p}_t$, while for $t > t_c$, $\vec{p}_t \prec_{\vec{\pi}} \vec{q}_t$.

    The probability $\vec{p}$ thermomajorizes $\vec{q}$, $\vec{q} \prec_{\vec{\pi}} \vec{p}$, if and only if there exists a stochastic matrix $\hat{S}$ such that $\vec{q} = \hat{S} \vec{p}$ and $\vec{\pi} = \hat{S} \vec{\pi}$.
    Here, the stochastic matrix $\hat{S}$ has non-negative elements and satisfies $\hat{S}^{T} \vec{I} = \vec{I}$, where $\vec{I} = (1,\dots,1)^T$. 
    This identifies thermomajorization as a special case of the operation Mpemba effect in classical thermodynamics, where thermal stochastic matrices with a fixed point $\vec{\pi}$ serve as the free operations.

    Denote the thermal relaxation dynamics as $\hat{R}_t$ and the thermal-Markovian dynamics as $\hat{S}_{\lambda}$. 
    If the operation Mpemba effect does not occur at the time $t$, the induced dynamics $\hat{R}_t \hat{S}_{\lambda}$ is thermal-divisible, i.e., the intermediate matrix $\hat{R}_t \hat{S} \hat{R}_t^{-1}$ is still thermal.
    Thus, for $\vec{q} = \hat{S}_{\lambda} \vec{p}$ one has $\vec{q}_t = \hat{R}_t \hat{S} \hat{R}_t^{-1}\vec{p}_t$, which implies that $\vec{q}_t \prec_{\vec{\pi}} \vec{p}_t$ and excludes the thermomajorization Mpemba effect. 
    Therefore, operation Mpemba effect is necessary for thermomajorization Mpemba effect.
    The converse is not automatic: even if $\hat{R}_t \hat{S} \hat{R}_t^{-1}$ is athermal, the inverse intermediate matrix $\hat{R}_t \hat{S}^{-1} \hat{R}_t^{-1}$ need not be thermal, so the thermomajorization Mpemba effect may fail to occur.
    
    Then, we consider a three-level system $\vec{p} = (p_0,p_1,p_2)^T$ as an example.
    The thermal relaxation dynamics is governed by the master equation ${\mathrm{d} \vec{p}}/{\mathrm{d}t}  = \hat{W} \vec{p}$, where the rate matrix $\hat{W} = [w_{nm}]$ satisfies the detailed balance condition $w_{nm} \pi_m = w_{mn} \pi_n$ for $n \neq m$.
    The off-diagonal elements $w_{nm} \geq 0$ for $n \neq m$, and the diagonal elements $w_{nn} = - \sum_{m \neq n} w_{mn}$.
    For simplicity, we set the thermal state at infinite temperature as $\vec{\pi} = (1/3,1/3,1/3)^T$, so the rate matrix $\hat{W}$ is symmetric.
    The initialization dynamics is governed by the master equation ${\mathrm{d}\vec{p}}/{\mathrm{d}\lambda}  = \hat{A} \vec{p}$ with respect to the parameter $\lambda$, where $\hat{A} = [a_{nm}]$.
    
    Figure~\ref{figure3} demonstrates the thermomajorization Mpemba effect with the detector $D(\vec{p},\vec{\pi}) = \max_n [p_n/\pi_n]$ of thermomajorization~\cite{PhysRevLett.134.107101}.
    In Fig.~\ref{figure3}(a), the detector $D_{\lambda,t} = D(\hat{R}_t \hat{S}_{\lambda} \vec{p},\vec{\pi})$ exhibits the Mpemba effect between different initialization parameters $\lambda$, after $t > 17$.
    In Fig.~\ref{figure3}(b), the increase of $L_1(\lambda,t) = \Vert \hat{R}_t \hat{S}_{\lambda} \hat{R}_t^{-1} \Vert_1$, which witnesses the athermality of the intermediate matrix, shows that the operation Mpemba effect occur at $t \approx 5$, indicating that the operational signatures appear earlier than the state signatures. 
    In Fig.~\ref{figure3}(c), the nonzero total positive change $N_{t}$ and normalized total positive change $N_{t}/D_{0,t}$ of the estimator $D_{\lambda,t}$ witness the occurrence of the thermomajorization Mpemba effect at $t \approx 17$.
    See Supplemental Material~\ref{app: thermal}, for more details.

    State convertibility for free operations induces a partial order on states in resource theories, and thermomajorization is a realization in classical thermodynamics. 
    In reversible resource theories with one-shot free convertibility, the order is governed by a unique resource measure and becomes a total order. 
    The operation Mpemba effect only breaks this order-preserving structure, whereas the thermomajorization Mpemba effect reverses the order between the evolved states. 
    For a total order, breaking and reversing the order coincide, while for a partial order the latter is strictly stronger. 
    In this operational framework, we can define a stronger version of the Mpemba effect by requiring the reversal of free convertibility order, which can be viewed as a direct generalization of the thermomajorization Mpemba effect to general resource theories.
    However, our definition of operation Mpemba effect, breaking the resource-Markovianity, is more concise and also captures many interesting implications of the Mpemba effect.

    \emph{Discussion}---%
    We have introduced the operation Mpemba effect in general resource theories as the breakdown of resource-Markovianity of free initialization operations. 
    This formulation is a special case of the state-based formulation, which connects the Mpemba effect with non-Markovianity of free dynamics.
    The Mpemba effect for distinguishability and thermomajorization are demonstrated from the perspective of the operation Mpemba effect.

    Our framework suggests several directions for future work. 
    First, it provides a measure-independent characterization of Mpemba effect in general resource theories, enabling the exploration of anomalous relaxation beyond thermodynamics in settings such as entanglement, coherence, and asymmetry.
    Second, by relating Mpemba effect to non-Markovianity of free dynamics, it introduces quantitative witnesses of Mpemba effect through the resource backflow, thereby broadening the scope of both the Mpemba effect and non-Markovianity.
    Finally, shifting the focus from isolated initial-state comparisons to initialization families, the operational formulation may guide initialization engineering in open quantum systems, such as rapid state preparation or engineered thermalization.
    
    \emph{Acknowledgments}---%
    This work was supported by the National Natural Science Foundation of China (Grants No.~T2121001, No.~U25A6009, No.~92365301, No.~12475017), the MOST (Grant No.~2025YFE0217600), the Scientific Research Innovation Capability Support Project for Young Faculty (Grant No.~SRICSPYF-ZY2025171), the Natural Science Foundation of Guangdong Province (Grant No.~2024A1515010398), the Guangdong Provincial Quantum Science Strategic Initiative (Grant No.~GDZX2505004),  and the Startup Grant of South China University of Technology (Grant No.~20240061). 


%

\clearpage

\onecolumngrid
\hypersetup{pageanchor=false}
\setcounter{page}{1}
\setcounter{section}{0}
\setcounter{equation}{0}
\setcounter{figure}{0}
\setcounter{table}{0}
\setcounter{proposition}{0}
\renewcommand{\thesection}{\Roman{section}}
\renewcommand{\theequation}{S\arabic{equation}}
\renewcommand{\thefigure}{S\arabic{figure}}
\renewcommand{\thetable}{S\arabic{table}}
\renewcommand{\theproposition}{S\arabic{proposition}}
\renewcommand{\theHsection}{supplement.\Roman{section}}
\renewcommand{\theHequation}{supplement.S\arabic{equation}}
\renewcommand{\theHfigure}{supplement.S\arabic{figure}}
\renewcommand{\theHtable}{supplement.S\arabic{table}}
\renewcommand{\theHproposition}{supplement.S\arabic{proposition}}

\begin{center}
    {\large\bfseries Supplemental Material for}\\[0.5em]
    {\large\bfseries ``\papertitle''}\\[1.0em]
    Tian-Ren Jin$^{1,2}$, Yu-Ran Zhang$^{3}$, and Heng Fan$^{1,2,4,5,6,7}$\\[0.5em]
    {\small \it
    $^{1}$Institute of Physics, Chinese Academy of Sciences, Beijing 100190, China\\
    $^{2}$School of Physical Sciences, University of Chinese Academy of Sciences, Beijing 100049, China\\
    $^{3}$School of Physics and Optoelectronics, South China University of Technology, Guangzhou 510640, China\\
    $^{4}$Beijing Academy of Quantum Information Sciences, Beijing 100193, China\\
    $^{5}$Hefei National Laboratory, Hefei 230088, China\\
    $^{6}$Songshan Lake Materials Laboratory, Dongguan 523808, China\\
    $^{7}$CAS Center for Excellence in Topological Quantum Computation, UCAS, Beijing 100190, China
    }
\end{center}

\begin{center}
\begin{minipage}{\textwidth}
    \large
    \noindent\textbf{Contents}\\[1em]
    \noindent\makebox[2.6em][r]{\hyperref[app: reversibility]{I.\ }}\hyperref[app: reversibility]{State-based and Operation Mpemba Effects in Reversible Resource Theories}\hfill\pageref{app: reversibility}\\[0.5em]
    \noindent\makebox[2.6em][r]{\hyperref[app: master]{II.\ }}\hyperref[app: master]{Operation Mpemba Effects in Lindblad Form}\hfill\pageref{app: master}\\[0.5em]
    \noindent\makebox[2.6em][r]{\hyperref[app: distinguishability]{III.\ }}\hyperref[app: distinguishability]{Details of Mpemba Effect for Distinguishability}\hfill\pageref{app: distinguishability}\\[0.5em]
    \noindent\makebox[2.6em][r]{\hyperref[app: thermal]{IV.\ }}\hyperref[app: thermal]{Details of Thermomajorization Mpemba Effect}\hfill\pageref{app: thermal}
\end{minipage}
\end{center}
\vspace{1em}

\section{State-based and Operation Mpemba Effects in Reversible Resource Theories}
\label{app: reversibility}

    This supplementary material explains the relationship between state-based and operation Mpemba effects in reversible resource theories. 
    In a reversible resource theory, there exists a unique asymptotic resource measure $M^{\infty}(\rho)$ governing the optimal asymptotic conversion rate between arbitrary states. 
    More precisely, for any two states $\rho$ and $\sigma$, there exists a sequence of free operations $\mathcal{F}_n$ and integers $m_n$ such that
    \begin{equation}
        \lim_{n\to\infty}\left\|\mathcal{F}_n\bigl(\rho^{\otimes n}\bigr)-\sigma^{\otimes m_n}\right\|_1 = 0,
    \end{equation}
    and
    \begin{equation}
        \lim_{n\to\infty}\frac{m_n}{n} = \frac{M^{\infty}(\rho)}{M^{\infty}(\sigma)}.
    \end{equation}
    In particular, if $M^{\infty}(\rho)\geq M^{\infty}(\sigma)$, then $m_n\geq n$ for sufficiently large $n$. 
    Since discarding subsystems is a free operation, one may first asymptotically convert $\rho^{\otimes n}$ into $\sigma^{\otimes m_n}$ and then discard $m_n-n$ replicas, thereby obtaining $\sigma^{\otimes n}$ asymptotically. 
    Thus, the states are totally ordered by asymptotic free convertibility, and this order is governed by the unique asymptotic resource measure $M^{\infty}$. 
    This total order need not persist at the one-shot level. 
    When it does, i.e., when for two states $\rho$ and $\sigma$ with $M^{\infty}(\rho)\geq M^{\infty}(\sigma)$ there exists a free operation $\mathcal{F}$ such that $\sigma=\mathcal{F}(\rho)$, we say that the reversible resource theory admits one-shot free convertibility.

    \begin{proposition}
        In reversible resource theories, the operation Mpemba effect and the state-based Mpemba effect are asymptotically equivalent at any finite parameter resolution.
    \end{proposition}

    \begin{proof}
        The implication from the operational formulation to the state-based formulation is the proposition in the main text and does not require reversibility. 
        Conversely, assume that the state-based Mpemba effect for the relaxation operation $\mathcal{N}_t$ holds for a resource-monotonic state family $\rho_{\lambda}$.
        Consider a finite discretization $\lambda^{(1)}<\lambda^{(2)}<\cdots<\lambda^{(L)}$ of the resource-monotonic state family $\rho_{\lambda}$. 
        For each adjacent pair $\lambda^{(j)}<\lambda^{(j+1)}$, asymptotic reversibility gives a sequence of free operations $\mathcal{E}^{(n)}_{\lambda^{(j+1)},\lambda^{(j)}}$ such that
        \begin{equation}
            \lim_{n\to\infty}\left\| \mathcal{E}^{(n)}_{\lambda^{(j+1)},\lambda^{(j)}} \bigl(\rho_{\lambda^{(j)}}^{\otimes n}\bigr) - \rho_{\lambda^{(j+1)}}^{\otimes n} \right\|_1 =0.
        \end{equation}
        Define the discretized initialization family from the common input $\rho_{\lambda^{(1)}}^{\otimes n}$ by
        \begin{equation}
            \mathcal{E}^{(n)}_{\lambda^{(1)}}=\mathcal{I}^{\otimes n},\qquad
            \mathcal{E}^{(n)}_{\lambda^{(j)}}= \overleftarrow{\bigcirc}_{m=1}^{j-1} \mathcal{E}^{(n)}_{\lambda^{(m+1)},\lambda^{(m)}}, 
        \end{equation}
        where $\overleftarrow{\bigcirc}$ denotes composition from right to left as $m$ increases.
        Since each intermediate map is free, this discretized family is resource-Markovian. 
        Moreover, by repeated use of the triangle inequality and contractivity of the trace distance under quantum operations, one obtains
        \begin{align}
            \Vert \mathcal{E}^{(n)}_{\lambda^{(j)}}\bigl(\rho_{\lambda^{(1)}}^{\otimes n}\bigr)-\rho_{\lambda^{(j)}}^{\otimes n}\Vert_1 
            & \leq \Vert \mathcal{E}^{(n)}_{\lambda^{(j)}}(\rho_{\lambda^{(1)}}^{\otimes n}) - \mathcal{E}^{(n)}_{\lambda^{(j)}}(\rho_{\lambda^{(j-1)}}^{\otimes n})\Vert_1 + \Vert  \mathcal{E}^{(n)}_{\lambda^{(j)}}(\rho_{\lambda^{(j-1)}}^{\otimes n}) - \rho_{\lambda^{(j)}}^{\otimes n}\Vert_1 \nonumber \\
            & \leq \Vert \rho_{\lambda^{(1)}}^{\otimes n} - \rho_{\lambda^{(j-1)}}^{\otimes n}\Vert_1 + \Vert  \mathcal{E}^{(n)}_{\lambda^{(j)}}(\rho_{\lambda^{(j-1)}}^{\otimes n}) - \rho_{\lambda^{(j)}}^{\otimes n}\Vert_1 \nonumber \\
            & \leq \sum_{m=1}^{j-1} \Vert  \mathcal{E}^{(n)}_{\lambda^{(m+1)},\lambda^{(m)}}(\rho_{\lambda^{(m)}}^{\otimes n}) - \rho_{\lambda^{(m+1)}}^{\otimes n}\Vert_1 .
        \end{align}
        Taking the limit $n\to\infty$ yields
        \begin{equation}
            \lim_{n\to\infty} \Vert \mathcal{E}^{(n)}_{\lambda^{(j)}}\bigl(\rho_{\lambda^{(1)}}^{\otimes n}\bigr)-\rho_{\lambda^{(j)}}^{\otimes n}\Vert_1 = 0,
        \end{equation}
        for all $j=1,\dots,L$.
        Thus, the discretized resource-monotonic state family is asymptotically initialized by the finite-resolution resource-Markovian dynamics $\mathcal{E}^{(n)}_{\lambda^{(j)}}$ from the common input $\rho_{\lambda^{(1)}}^{\otimes n}$.
        If the discretized state family exhibits a state-based Mpemba effect under the replica-wise relaxation dynamics $\mathcal{N}_t^{\otimes n}$, then the induced discretized initialization family is no longer resource-Markovian after this relaxation.
        Hence the operation Mpemba effect also occurs asymptotically on the finite discretization.
    \end{proof}

    Extending this finite-resolution statement to a fully continuous asymptotic equivalence requires constructing a continuous resource-Markovian dynamics that asymptotically initializes the resource-monotonic state family $\rho_{\lambda}$.
    Such a dynamics can be constructed either globally or locally.
    A global construction groups all free operations $\mathcal{E}^{(n)}_{\lambda,0}$ that asymptotically convert the common input $\rho_{0}^{\otimes n}$ to each state $\rho_{\lambda}^{\otimes n}$, thus asymptotically initializes the whole family $\rho_{\lambda}^{\otimes n}$ automatically.
    However, the resource-Markovianity of the constructed dynamics has no guarantee.
    A local construction composes the infinitesimal free intermediate operations $\mathcal{E}^{(n)}_{\lambda+\mathrm{d}\lambda,\lambda}$ that asymptotically convert $\rho_{\lambda}^{\otimes n}$ to $\rho_{\lambda+\mathrm{d}\lambda}^{\otimes n}$ for infinitesimal $\mathrm{d}\lambda$, which automatically satisfies resource-Markovianity.
    In this case, the asymptotically initializing property requires that 
    \begin{equation}
        \lim_{n\to\infty} \int_0^{\lambda} f_n(\tilde{\lambda}) \mathrm{d}\tilde{\lambda} = 0,
    \end{equation} 
    holds for all $\lambda$, where $f_n(\tilde{\lambda}) = \lim_{\mathrm{d}\tilde{\lambda} \to 0} \Vert \mathcal{E}^{(n)}_{\tilde{\lambda} + \mathrm{d}\tilde{\lambda},\tilde{\lambda}}(\rho_{\tilde{\lambda}}^{\otimes n}) - \rho_{\tilde{\lambda} + \mathrm{d}\tilde{\lambda}}^{\otimes n}\Vert_1/\mathrm{d}\tilde{\lambda}$.
    Notice that in the limit $n\to\infty$, the function $f_n(\tilde{\lambda})$ converges pointwise to zero for all $\tilde{\lambda}$.
    By Vitali's convergence theorem, convergence of the integrals is equivalent to uniform integrability of the function sequence $\{f_n(\tilde{\lambda})\}$, which is a stronger condition than pointwise convergence and does not automatically follow from asymptotic reversibility alone.
    Therefore, the asymptotic equivalence between the operational and state-based Mpemba effects cannot be established in the continuous sense without constraints on the convergence rate of asymptotic reversibility.

    \begin{proposition}
        If the resource theory admits one-shot free convertibility, then the operation Mpemba effect and the state-based Mpemba effect are exactly equivalent.
    \end{proposition}

    \begin{proof}
        Again, the implication from the operational formulation to the state-based formulation follows from the proposition in the main text.
        Conversely, assume that the state-based Mpemba effect for the relaxation operation $\mathcal{N}_t$ holds for a resource-monotonic state family $\rho_{\lambda}$.
        By one-shot free convertibility, for each pair of parameters $(\lambda_1, \lambda_2)$, there exists a free operation $\mathcal{E}_{\lambda_2,\lambda_1}$ such that
        \begin{equation}
            \rho_{\lambda_2} = \mathcal{E}_{\lambda_2,\lambda_1}(\rho_{\lambda_1}).
        \end{equation}
        Now consider the local construction of the free initialization dynamics by composing the infinitesimal free intermediate operations $\mathcal{E}_{\lambda+\mathrm{d}\lambda,\lambda}$, which is resource-Markovian by construction.
        Since each infinitesimal step exactly converts the state $\rho_{\lambda}$ to $\rho_{\lambda+\mathrm{d}\lambda}$, the constructed dynamics exactly initializes the whole resource-monotonic state family $\rho_{\lambda}$ from the common input $\rho_{0}$.
        This yields an exact resource-Markovian operational realization of the resource-monotonic state family.
        Hence, the two formulations are exactly equivalent.
    \end{proof}

\section{Operation Mpemba Effects in Lindblad Form}
\label{app: master}

    This supplementary material applies to resource theories whose free operations are generated by a closed convex cone of Lindblad generators. In that setting, the free set is closed under composition and forms a monoid (a semigroup with identity), and the continuously connected free operations can be represented as exponentials $\exp \mathcal{L}$ of Lindblad generators~\cite{shibata1977generalized,chaturvedi1979time}:
    \begin{align}
    \mathcal{L}(\cdot) = - \mathrm{i} [\hat{H}, (\cdot)] 
    + \sum_k \gamma_{k} \left[\hat{A}_{k} (\cdot) \hat{A}_{k}^{\dagger} - \frac{1}{2} \{\hat{A}_{k}^{\dagger}\hat{A}_{k}, (\cdot)\}\right].
    \end{align}

    Assume the generators of the free set are $\mathcal{L}_i$. By Trotter decomposition, the operation
    \begin{equation}
    \exp\left(\sum_i x_i \mathcal{L}_i\right) = \lim_{N\rightarrow \infty} \left(\prod_i e^{x_i \mathcal{L}_i/N}\right)^N,
    \end{equation}
    is asymptotically free when $x_i \geq 0$. We assume the free set is closed, so it is identified with a closed convex cone in the space spanned by these Lindblad generators.

    The adjoint action of the operation $\mathcal{N}$ on free operations is
    \begin{equation}
    \mathcal{N} \circ \exp\left(\sum_i x_i \mathcal{L}_i\right) \circ \mathcal{N}^{-1} = \exp\left(\sum_i x_i \tilde{\mathcal{L}}_i\right),
    \end{equation}
    where $\tilde{\mathcal{L}}_i = \mathcal{N} \circ \mathcal{L}_i\circ \mathcal{N}^{-1}$. If the operation $\mathcal{N}$ has no operation Mpemba effect, the operations $\tilde{\mathcal{L}}_i$ are free, i.e., there are $\alpha_{ij}(\mathcal{N}) \geq 0$ such that
    \begin{equation}
    \mathcal{N} \circ \mathcal{L}_i\circ \mathcal{N}^{-1} = \sum_j \alpha_{ij}(\mathcal{N}) \mathcal{L}_j,
    \end{equation}
    and the elements $\alpha_{ij}(\mathcal{N})$ of the adjoint representation of $\mathcal{N}$ on the space spanned by the generators $\mathcal{L}_i$ of the free set are all non-negative.

    Moreover, the relaxation operation $\mathcal{N}$ is also free. Denote $\mathcal{L}_{\mathcal{N}}$ as the generator of $\mathcal{N} = \exp \mathcal{L}_{\mathcal{N}}$, which is expressed as $\mathcal{L}_{\mathcal{N}} = \sum_i n_i \mathcal{L}_{i}$ with $n_i \geq 0$. If the generators $\mathcal{L}_i$ of the free set are closed under the Lie product
    \begin{equation}
    [\mathcal{L}_i,\mathcal{L}_j] = \sum_k \beta_{jk}(\mathcal{L}_i) \mathcal{L}_k,
    \end{equation}
    where $\beta_{jk}(\mathcal{L}_i) \geq 0$ are non-negative, then the resource theory has no free operation exhibiting operation Mpemba effect, since
    \begin{equation}
    \mathcal{N} \circ \mathcal{L}_i\circ \mathcal{N}^{-1} = \sum_{j} \left[\exp\left(\sum_k n_k \beta(\mathcal{L}_k)\right)\right]_{ij} \mathcal{L}_j,
    \end{equation}
    where the matrix elements $\left[\exp\left(\sum_k n_k \beta(\mathcal{L}_k)\right)\right]_{ij}$ are all non-negative due to $n_k$ and $\beta_{ij}(\mathcal{L}_k)$ being non-negative.

    \section{Details of Mpemba Effect for Distinguishability}
    \label{app: distinguishability}

    In this supplementary material, we provide details about the Mpemba effect for distinguishability. 
    In Bloch vector, the Lindblad equation can be represented as $\frac{\mathrm{d} \vec{r}}{\mathrm{d} t} = \Omega \vec{r} + \vec{b}$, where
    \begin{equation}
        \Omega_{ij} = \left(\begin{array}{ccc} 
            -C_{22}-C_{33} & -h_3 + \mathrm{Re} C_{12} & h_2 + \mathrm{Re} C_{31} \\ 
            h_3 + \mathrm{Re} C_{21} & -C_{11}-C_{33} & -h_1 + \mathrm{Re} C_{23} \\ 
            -h_2 + \mathrm{Re} C_{31} & h_1 + \mathrm{Re} C_{32} & -C_{11}-C_{22} 
        \end{array}\right),
    \end{equation}
    and $\vec{b} = -4 (\mathrm{Im} C_{23}, \mathrm{Im} C_{31}, \mathrm{Im} C_{12})^{T}$.
    Since the trace distance only depends on the difference $\vec{r}^{(12)} = \vec{r}^{(1)} - \vec{r}^{(2)}$, which satisfies $\frac{\mathrm{d}\vec{r}^{(12)}}{\mathrm{d} t} = \Omega \vec{r}^{(12)}$, we assume $\vec{b} = 0$, and
    \begin{equation}
        \Omega_{ij} = \left(\begin{array}{ccc} 
            -C_{22}-C_{33} & -h_3 + C_{12} & h_2 + C_{31} \\ 
            h_3 +  C_{21} & -C_{11}-C_{33} & -h_1 + C_{23} \\ 
            -h_2 +  C_{31} & h_1 +  C_{32} & -C_{11}-C_{22} 
        \end{array}\right).
    \end{equation}
    Similarly, the Lindblad matrix of initialization dynamics is 
    \begin{equation}
        \Theta_{ij} = \left(\begin{array}{ccc} 
            -D_{22}-D_{33} & -k_3 + D_{12} & k_2 + D_{31} \\ 
            k_3 + D_{21} & -D_{11}-D_{33} & -k_1 + D_{23} \\ 
            -k_2 + D_{31} & k_1 + D_{32} & -D_{11}-D_{22} 
        \end{array}\right).
    \end{equation}

    With the parameterization of the initialization dynamics $\mathcal{E}_{\lambda}$ and the relaxation operation $\mathcal{N}_t$, the output Bloch vector after the induced operation $\mathcal{N}_t \circ \mathcal{E}_{\lambda}$ is
    \begin{equation}
        \vec{r}^{(12)}(\lambda,t) = \exp(\Omega t) \exp(\Theta \lambda) \vec{r}^{(12)}(0).
    \end{equation}
    The rate of change of the trace distance is
    \begin{align}
        \sigma(\lambda) & = \frac{\mathrm{d}}{\mathrm{d} \lambda} D(\vec{r}^{(12)}(\lambda,t)) 
        = \frac{\vec{r}^{(12)}(\lambda,t)^{T} \exp(\Omega t)\Theta\exp(-\Omega t) \vec{r}^{(12)}(\lambda,t)}{|\vec{r}^{(12)}(\lambda,t)|}.
    \end{align}
    Under the constraint $r_3^{(12)} = 0$, the dynamics of initialization $\exp(\Theta \lambda)$ and relaxation $\exp(\Omega t)$ are effectively in two-dimension, which can be represented as $\Theta = \vec{\Theta} \cdot \hat{\bm{\sigma}} + \Theta_0$ with $\vec{\Theta} = (\Theta_1, \Theta_2, \Theta_3)^T$, where
    \begin{align}
        \Theta_0 & = - D_{33} - \frac{1}{2}(D_{11}+D_{22}), \quad
        \Theta_1 = D_{12}, \\
        \Theta_2 & = -\mathrm{i} k_3, \quad
        \Theta_3 = \frac{1}{2}(D_{11}-D_{22}),
    \end{align} 
    and $\Omega = \vec{\Omega} \cdot \hat{\bm{\sigma}} + \Omega_0$ with $\vec{\Omega} = (\Omega_1, \Omega_2, \Omega_3)^T$, where
    \begin{align}
        \Omega_0 & = - C_{33} - \frac{1}{2}(C_{11}+C_{22}), \quad
        \Omega_1 = C_{12}, \\
        \Omega_2 & = -\mathrm{i} h_3, \quad
        \Omega_3 = \frac{1}{2}(C_{11}-C_{22}).
    \end{align} 
    The rate of change for $\vec{r}^{(12)} = (r_1, r_2)^T$ is
    \begin{equation}
        \sigma = \frac{1}{|\vec{r}^{(12)}|} \vec{r}^{(12)T} \left[\Theta_1 \hat{\sigma}_1 + \Theta_3 \hat{\sigma}_3\right] \vec{r}^{(12)} + \Theta_0 |\vec{r}^{(12)}|.
    \end{equation}
    For a Markovian dynamics, the rate of change of the trace distance $\sigma$ is non-positive for any vector $\vec{r}^{(12)}$, which is equivalent to 
    \begin{equation}
        \Theta_1 \hat{\sigma}_1 + \Theta_3 \hat{\sigma}_3 \leq \Theta_0, \quad \mathrm{or} \quad (\Theta_1^2 + \Theta_3^2)^{1/2} \leq \Theta_0.
    \end{equation}

    By setting $\Omega_2 = 0$, the relaxation dynamics yields
    \begin{align}
        & (\Theta_1^2 + \Theta_3^2)^{1/2} = |\vec{l} \cdot \vec{\Theta}| \sqrt{1 - (\vec{e_2} \cdot\vec{l})^2/|\vec{l}|^2} 
        \exp \left[\sqrt{1 - (\vec{e_2} \cdot\vec{l})^2/|\vec{l}|^2} t/2\right] + O(1),
    \end{align}
    where $\vec{e_2} = (0,1,0)^T$, $\vec{l} = \left(-2 \Omega_3, -2\mathrm{i} \sqrt{\Omega_1^2 + \Omega_3^2}, 2\Omega_1\right)^T$ is the eigenvector with positive eigenvalue. 
    Note that
    \begin{equation}
        \sqrt{1 - (\vec{e_2} \cdot\vec{l})^2/|\vec{l}|^2} = 2 \sqrt{\Omega_1^2 + \Omega_3^2} \neq 0,
    \end{equation}
    otherwise the operation $\Omega$ is trivial and no Mpemba effect occurs. 
    The quantity $(\Theta_1^2 + \Theta_3^2)^{1/2}$ increases exponentially if $\vec{l} \cdot \vec{\Theta} \neq 0$, and the Mpemba effect for distinguishability occurs for some time.

    Specifically, we consider the initial Markovian dynamics with $k_3 = 4$, $D_{11} = D_{22} = D_{12} = 0$, and $D_{33} = 1$, and the relaxation operation with $h_3 = 0$, $C_{11} = C_{22} = C_{12} = C_{33} = 1$. 
    Then, $\vec{l} \cdot \vec{\Theta} = 8 \neq 0$. 
    We take the two initial states as $\ket{+} = \frac{1}{\sqrt{2}}(\ket{0} + \ket{1})$ and $\ket{-} = \frac{1}{\sqrt{2}}(\ket{0} - \ket{1})$, thus $\vec{r}^{(12)}(0) = (2,0)^T$.
    The dynamics of the initialization operation is
    \begin{equation}
        \mathcal{E}_{\lambda} = \exp(\Theta \lambda) = e^{-\lambda} \left(\begin{array}{cc} \cos4\lambda & -\sin4\lambda \\ \sin4\lambda & \cos4\lambda \end{array}\right),
    \end{equation}
    and the dynamics of the relaxation operation is
    \begin{equation}
        \mathcal{N}_t = \exp(\Omega t) = e^{-2t} \left(\begin{array}{cc} \cosh t & \sinh t \\ \sinh t & \cosh t \end{array}\right).
    \end{equation}
    When $\lambda = 0$, the trace distance after $\mathcal{N}_t \circ \mathcal{E}_{\lambda}$ is
    \begin{equation}
        D_{\mathcal{N}_t \circ \mathcal{E}_{0}}(\ket{+},\ket{-}) = e^{- 2 t} \cosh^{1/2} 2t = e^{- t} \left(\frac{1+e^{-4t}}{2}\right)^{1/2},
    \end{equation}
    whereas when $\lambda = \frac{\pi}{16}$, the trace distance is
    \begin{equation}
        D_{\mathcal{N}_t \circ \mathcal{E}_{\frac{\pi}{16}}}(\ket{+},\ket{-}) = e^{-t-\pi/16}.
    \end{equation}
    Initially, $D_{\mathcal{N}_t \circ \mathcal{E}_{0}}(\ket{+},\ket{-}) > D_{\mathcal{N}_t \circ \mathcal{E}_{\frac{\pi}{16}}}(\ket{+},\ket{-})$.
    For long time, $D_{\mathcal{N}_t \circ \mathcal{E}_{0}}(\ket{+},\ket{-}) \approx \frac{1}{\sqrt{2}} e^{- t}$, which is smaller than $D_{\mathcal{N}_t \circ \mathcal{E}_{\frac{\pi}{16}}}(\ket{+},\ket{-})$. 
    The trace distance crossing occurs at
    \begin{equation}
        t_c = - \frac{1}{4}\ln(2 e^{-\pi/8} - 1) \approx 0.2621,
    \end{equation}
    The dynamics of trace distance is plotted in Fig.~\ref{figure2}(a).
    In Fig.~\ref{figure2}(b) and (c), the setting has varying $k_3$ while other parameters are the same.

\section{Details of Thermomajorization Mpemba Effect}
\label{app: thermal}

    The classical probability distribution $\vec{p} = (p_1,\dots,p_d)$ thermomajorizes $\vec{q} = (q_1,\dots,q_d)^T$ if the Lorenz curve of $(\vec{p},\vec{\pi})$ lies above that of $(\vec{q},\vec{\pi})$.
    The Lorenz curve of $(\vec{p},\vec{\pi})$ consists of the piecewise lines connecting the points $\left(\sum_{k=1}^{n} \tilde{\pi}_k, \sum_{k=1}^{n} \tilde{p}_k\right)$, for $n = 0, 1, \dots, d$, where the probability distribution $\tilde{\vec{p}} = (\tilde{p}_1,\dots, \tilde{p}_d)$ and $\tilde{\vec{\pi}} = (\tilde{\pi}_1,\dots, \tilde{\pi}_d)$ are the permuted probability of $\vec{p}$ and $\vec{\pi}$, such that $\tilde{p}_1/\tilde{\pi}_1 \geq \tilde{p}_2/\tilde{\pi}_2 \geq \dots \geq \tilde{p}_d/\tilde{\pi}_d$.
    Here, for $n=0$, the point is $(0,0)$.

    With the relation between thermomajorization and stochastic matrix, we illustrate the relation between thermomajorization Mpemba effect and operation Mpemba effect.
    If operation Mpemba effect does not occur at time $t$, the induced dynamics $\hat{R}_t \hat{S}_{\lambda}$ is thermal-divisible, and the intermediate matrix $\hat{R}_t \hat{S} \hat{R}_t^{-1}$ is still thermal. 
    Since $\vec{q} = \hat{S} \vec{p}$ initially, $\vec{q}_t = \hat{R}_t \hat{S} \hat{R}_t^{-1} \vec{p}_t$ implies that $\vec{q}_t \prec_{\vec{\pi}} \vec{p}_t$, and the thermomajorization Mpemba effect does not occur.
    Furthermore, if operation Mpemba effect occurs at time $t$, the intermediate matrix $\hat{R}_t \hat{S} \hat{R}_t^{-1}$ is no longer thermal.
    However, although $\hat{R}_t \hat{S} \hat{R}_t^{-1}$ is athermal,  $\hat{R}_t \hat{S}^{-1} \hat{R}_t^{-1}$ may also be athermal.
    With $\vec{q} = \hat{S} \vec{p}$, $\vec{p}_t = \hat{R}_t \hat{S}^{-1} \hat{R}_t^{-1} \vec{q}_t$ does not imply that $\vec{p}_t \prec_{\vec{\pi}} \vec{q}_t$. The thermomajorization Mpemba effect requires the inverse intermediate matrix to be thermal.

    Since both the relaxation dynamics $\hat{R}_t$ and the initialization dynamics $\hat{S}$ are thermal, the intermediate matrix $\hat{R}_t \hat{S} \hat{R}_t^{-1}$ satisfies 
    \begin{equation}
        (\hat{R}_t \hat{S} \hat{R}_t^{-1}) \vec{\pi} = \vec{\pi}, \quad (\hat{R}_t \hat{S} \hat{R}_t^{-1})^{T} \vec{I} = \vec{I}.   
    \end{equation}
    The operation Mpemba effect requires the intermediate matrix $\hat{R}_t \hat{S} \hat{R}_t^{-1}$ has negative elements.
    Therefore, we introduce the $L_1$-norm $\Vert \cdot \Vert_1$ to detect the negative elements.
    For a stochastic matrix $\hat{M}$, the $L_1$-norm is $\Vert \hat{M} \Vert_1 = \max_j \sum_i |M_{ij}| = 1$, while for a matrix with negative elements, $\Vert \hat{M} \Vert_1 > 1$.
    The thermomajorization Mpemba effect can be detected by the $f_{\alpha}$-divergence with $\alpha \rightarrow 0$, which is equivalent to evaluating the cross-over of $\max_n [p_n/\pi_n]$~\cite{PhysRevLett.134.107101}.

    For a two-level system $\vec{p} = (p_0,p_1)^T$, set the thermal state as $\vec{\pi} = \left(\frac{\rho}{1+\rho},\frac{1}{1+\rho}\right)^T$, where $\rho \geq 0$.
    Thus, a thermal dynamics matrix is
    \begin{equation}
        \hat{S} = \left(\begin{array}{cc}
            1-a & a \rho  \\
            a & 1-a \rho
        \end{array}\right),
    \end{equation}
    where $0 \leq a \leq 1$.
    At the time $t$, denote the thermal relaxation matrix as
    \begin{equation}
        \hat{R}_t = \left(\begin{array}{cc}
            1-b & b \rho \\
            b & 1- b \rho
        \end{array}\right).
    \end{equation}
    The relaxation matrix $\hat{R}_t$ and the initialization matrix $\hat{S}$ commute, thus the intermediate matrix $\hat{R}_t \hat{S} \hat{R}_t^{-1} = \hat{S}$ is still thermal.
    This implies that neither the operation Mpemba effect nor the thermomajorization Mpemba effect occurs.
    For this reason, we consider the three-level system in the main text.

\end{document}